\documentclass[11pt]{article}
\usepackage{fullpage}
\usepackage{graphicx}
\usepackage{delarray}
\usepackage{bbm}
\usepackage{amssymb,amsmath}
\usepackage{amsthm}
\usepackage{amsfonts}
\usepackage{algorithm}
\usepackage{algorithmic}
\usepackage{complexity}

\newcommand{\argmin}{\operatornamewithlimits{argmin}}

\usepackage[numbers,sort&compress,comma]{natbib}

\usepackage{titlesec}
\titlespacing*{\section}{0pt}{12pt}{3pt}
\titlespacing*{\subsection}{0pt}{12pt}{3pt}
\titleformat{\section}[hang]{\bfseries}{\thesection}{1em}{}[]
\titleformat*{\subsection}{\bfseries}
\titleformat*{\subsubsection}{\bfseries}
\titleformat*{\paragraph}{\bfseries}
\titleformat*{\subparagraph}{\bfseries}

\usepackage{hyperref}
\hypersetup{
 linktocpage=true,
 pdfborderstyle={/S/S/W 1},
 hyperindex=true,
 bookmarks=true,
 bookmarksopen=true,
 bookmarksnumbered=true,
}

\begin{document}

\newtheorem{definition}{Definition}
\newtheorem{coro}{Corollary}
\newtheorem{conj}{Conjecture}
\newtheorem{thm}{Theorem}
\newtheorem{prf}{Proof}
\newtheorem{lma}{Lemma}
\newtheorem{prp}{Proposition}
\renewcommand{\R}{\mathbb{R}}
\renewcommand{\E}{\mathbb{E}}
\renewcommand{\PP}{\mathbb{P}}
\newcommand{\I}{\mathbb{I}}
\newcommand{\lam}{\lambda}
\newcommand{\bs}{\boldsymbol}
\newcommand{\ph}{{\hat \pi}}
\newcommand{\xth}{\hat{x}_t}
\newcommand{\be}{\begin{equation}}
\newcommand{\ee}{\end{equation}}
\newcommand{\bea}{\begin{eqnarray}}
\newcommand{\eea}{\end{eqnarray}}
\newcommand{\bfl}{\begin{flalign}}
\newcommand{\efl}{\end{flalign}}
\newcommand{\bfc}{\begin{figure}\begin{center}}
\newcommand{\efc}{\end{center}\end{figure}}
\newcommand{\n}{\nonumber}
\newcommand{\dd}{\mathrm{d}}
\newcommand{\nin}{\noindent}
\newcommand{\mc}{\mathcal}
\newcommand{\ds}{\displaystyle}
\newcommand{\la}{\leftarrow}
\newcommand{\uga}{\underline{ga}}
\newcommand{\oga}{\overline{ga}}
\newcommand{\wt}{\widetilde}
\newcommand{\ch}{\mathrm{CH}}

\title{Randomized Minmax Regret for Combinatorial Optimization Under Uncertainty\footnote{Research supported in part by NASA ESTOs Advanced Information
System Technology (AIST) program under grant number NNX12H81G. Also supported by NSF grant 1029603 and ONR grant N00014-12-1-0033.}}
\author{Andrew Mastin\thanks{Laboratory for Information and Decision Systems, Department of Electrical Engineering and Computer Science, Massachusetts Institute of Technology, Cambridge, MA 02139, USA; {\tt mastin@mit.edu}},~~
Patrick Jaillet\thanks{Laboratory for Information and Decision Systems, Department of Electrical Engineering and Computer Science and Operations Research Department, Massachusetts Institute of Technology, Cambridge, MA 02139, USA; {\tt jaillet@mit.edu}},~~ Sang Chin\thanks{Draper Laboratory, 555 Technology Square, Cambridge, MA 02139; {\tt schin@draper.com}}
}
\maketitle
\begin{abstract}
The minmax regret problem for combinatorial optimization under uncertainty can be viewed as a zero-sum game played between an optimizing player and an adversary, where the optimizing player selects a solution and the adversary selects costs with the intention of maximizing the regret of the player. The existing minmax regret model considers only deterministic solutions/strategies, and minmax regret versions of most polynomial solvable problems are \NP-hard. In this paper, we consider a randomized model where the optimizing player selects a probability distribution (corresponding to a mixed strategy) over solutions and the adversary selects costs with knowledge of the player's distribution, but not its realization. We show that under this randomized model, the minmax regret version of any polynomial solvable combinatorial problem becomes polynomial solvable. This holds true for both the interval and discrete scenario representations of uncertainty. Using the randomized model, we show new proofs of existing approximation algorithms for the deterministic model based on primal-dual approaches. Finally, we prove that minmax regret problems are \NP-hard under general convex uncertainty.

~\\
~\\

\end{abstract}

\section{Introduction}
Many optimization applications involve cost coefficients that are not fully known. When distributional information on cost coefficients is available (e.g. from historical data or other estimates), stochastic programming is often an appropriate modeling choice \cite{birge97, shapiro09}. In other cases, costs may only be known to be contained in intervals (i.e. each cost has a known lower and upper bound), or to be a member of a finite set of scenarios, and one is more interested in worst-case performance. Robust optimization formulations are desirable here as they employ a minmax-type objective and do not require knowledge of cost distributions \cite{ky97, bertsim04, kasperski08}.

In a general robust optimization problem with cost uncertainty, one must select a set of items from some feasible \textit{solution set}, such that item costs are unknown but must be contained in a known \textit{uncertainty set}. Under the well known \textit{minmax} objective (also referred to as absolute robustness), the goal is to select a solution that gives the best upper bound on objective cost over all possible costs from the uncertainty set \cite{wald39}. That is, one must select the solution that, when item costs are chosen to maximize the cost of the selected solution, is minimum. Under the \textit{minmax regret} objective (sometimes called the robust deviation model), the goal is instead to select the solution that minimizes the maximum possible regret, defined as the difference between the cost of the selected solution and the optimal solution \cite{savage51}.

A problem under the minmax regret objective can be viewed as a two stage game. In the first stage, the optimizing player selects a deterministic solution. In the second stage, an adversary observes the selected solution and chooses costs from the uncertainty set with the intention of maximizing the player's regret. The goal of the optimizing player is thus to select a solution that least allows the adversary to generate regret. For both interval and discrete scenario representations of cost uncertainty, the minmax regret versions of most polynomial solvable problems are \NP-hard \cite{aissi09}. A variation on this model, first suggested by Bertsimas et al.~\cite{bert12} for minmax robust optimization, is to allow the optimizing player to select a probability distribution over solutions and require the adversary to select costs based only on knowledge of the players distribution, but not its realization. In this paper, we show that under this randomized model, the minmax regret version of any polynomial solvable 0-1 integer linear programming problem becomes polynomial solvable. This holds true for both the interval and discrete scenario representations of uncertainty. 

Our crucial observation is that the randomized model is the linear programming relaxation of the integer program for the deterministic model. This leads to some useful insights. First, the minmax expected regret in the randomized model is upper bounded by the minmax regret in the deterministic model. Next, the linear program formulation can be used to create an approximation algorithm for the deterministic problem. We show that existing approximation algorithms for deterministic minmax regret problems, which have been proved using combinatorial arguments, can in fact be derived using primal-dual methods \cite{aissi06, kasperski06}. Our analysis here leads to lower bounds on randomized minmax regret with respect to the deterministic minmax regret, effectively stating limits on the power of using randomization.

Given that the randomized model makes the minmax regret problem polynomial solvable for interval uncertainty and discrete scenario uncertainty, it is natural to ask if polynomial solvability remains in the presence of slightly more elaborate uncertainty sets. We show that for general convex uncertainty sets, however, that the mere maximum regret problem (rather than the full minmax regret problem) is \NP-hard. The deterministic and randomized minmax regret problems are at least as hard as the maximum regret problem, so these problems become \NP-hard under general convex uncertainty.

The paper is structured as follows. In the remainder of this section we review related work; Section \ref{defsec} introduces notation and definitions. Section \ref{discretesec} presents the analysis for discrete scenario uncertainty, with derivations of optimal strategies for the optimizing player and the adversary, as well as the primal-dual approximation algorithm. Section \ref{intervalsec} gives the same results for interval uncertainty. Section \ref{nphardsec} demonstrates \NP-hardness of minmax regret problems under general convex uncertainty.  A conclusion is given in Section \ref{conclusion}.

\subsection{Related Work}
One of the first studies of minmax regret from both an algorithmic and complexity perspective was that of Averbakh \cite{aver01}. He looked at the minmax regret version of the simple problem of selecting $k$ items out of $n$ total items where the cost of each item is uncertain, and the goal is to select the set of items with minimum total cost. For interval uncertainty, he derived a polynomial time algorithm based on interchange arguments. He demonstrated that for the discrete scenario representation of uncertainty, however, the minmax regret problem becomes \NP-hard, even for the case of only two scenarios. It is interesting to contrast these results with the case of general minmax regret linear programming, which as shown by Averbakh and Lebedev  \cite{aver05}, is \NP-hard for interval uncertainty but polynomial solvable for discrete scenario uncertainty.

Apart from the item selection problem, most polynomial solvable minmax regret combinatorial problems are \NP-hard, both for interval and discrete scenario uncertainty. This is true for the shortest path, minimum spanning tree, assignment, and minimum $s$-$t$ cut problems \cite{yu98, ky97, aver04, aissi05, aissi05cut}. One exception is the minimum cut problem, the minmax regret version of which is polynomial solvable both for interval and discrete scenario uncertainty \cite{aissi05cut}. The survey paper of Aissi et al.~\cite{aissi09} provides a comprehensive summary of results related to both minmax and minmax regret combinatorial problems. For problems that are already \NP-complete, most of their minmax regret versions are $\Sigma_2^p$-complete (meaning that they are at the second level of the polynomial hierarchy) \cite{johannes12}. To solve minmax regret problems in practice, the book by Kasperski reviews standard mixed integer program (MIP) formulations for both interval and discrete scenario uncertainty \cite{kasperski08}. General approximation algorithms are known for both types of uncertainty. Kasperski and Zieli\'nski \cite{kasperski06} proved a general $2$-approximation algorithm based on midpoint costs under interval uncertainty, and Aissi et al. \cite{aissi06} gave a $k$-approximation algorithm using average costs under discrete scenario uncertainty, where $k$ is the number of scenarios.

The application of a game theoretic model with mixed strategies to robust optimization problems was introduced by Bertsimas et al. \cite{bert12}. They focused on the minmax robust model, and their analysis was motivated by adversarial models used for online optimization algorithms. As described by Ben-David et al.~\cite{bendavid94} (see also Borodin and El-Yaniv \cite{be98}), the three types of adversaries are the \textit{oblivious adversary}, the \textit{adaptive online adversary}, and the \textit{adaptive offline adversary}. The adaptive offline adversary is the analog of the conventional deterministic minmax regret problem, while the adaptive online adversary corresponds to our randomized model. The analog of the oblivious adversary, which we do study, is the model where the adversary first selects costs, and the optimizing player then selects the solution after viewing these costs.

For the randomized (corresponding to the adaptive online adversary) minmax problem, Bertsimas et al.~\cite{bert12} showed that if it is possible to optimize over both the solution set and the uncertainty set in polynomial time, then an optimal mixed strategy solution can be calculated in polynomial time, and that the expected cost under the randomized model is no greater than the cost for the deterministic model. This holds despite the fact that solving the minmax version of many polynomial solvable problems is \NP-hard for the deterministic case \cite{bertsimtech04}. They also gave lower bounds on the improvement gained from randomization for various uncertainty sets. Our work is similar to theirs, but we focus on the minmax regret objective instead of the minmax objective.

Another line of research that is related to ours is in security applications, where the adversarial model is realistically motivated. Korzhyk, et al.~\cite{korzhyk10} considered assignment-type problems where defensive resources, such as security guards, must be assigned to valued targets. They followed a Stackelberg model where the defending player has the power to commit to a mixed strategy; the attacker then observes this mixed strategy (though not the realization) and decides which targets to attack. They used linear programming formulations along with the Birkhoff-von Neumann theorem to find polynomial-sized optimal mixed strategies. It is also worth mentioning the work of Bertismas et al. on randomized strategies for network interdiction \cite{bert12network}.

\section{Definitions}
\label{defsec}
We consider a general combinatorial optimization problem where we are given a set of $n$ items $E= \{e_1, e_2,\ldots,e_n\}$ and a set $\mc{F}$ of feasible subsets of $E$. Each item $e \in E$ has a cost $c_e \in \R$. Given the vector $c = (c_1,\ldots,c_n)$, the goal of the optimization problem is to select the feasible subset of items that minimizes the total cost; we refer to this as the \textit{nominal problem}:
\bea
F^*(c) := \min_{T \in \mc{F}} \sum_{e \in T} c_e.
\eea
Let $x= (x_1,\ldots,x_n)$ be a characteristic vector for some set $T$, so that $x_e = 1$ if $e \in T$ and $x_e = 0$ otherwise. Also let $\mc{X} \subseteq \{0,1\}^n$ denote the set of all characteristic vectors corresponding to feasible sets $T \in \mc{F}$. We assume that $\mc{X}$ is described in size $m$ (e.g. with $m$ linear inequalities). We can equivalently write the nominal problem with a linear objective function:
\bea
F^*(c) = \min_{x \in \mc{X}} \sum_{e \in E} c_e x_e.
\eea
Throughout the paper, we will use both set notation and characteristic vectors for ease of presentation.
 
We will review the conventional regret definitions for the deterministic minmax regret framework, and then present the analogous definitions for our randomized model. For some cost vector $c \in \mc{C}$, the deterministic cost of a solution $T \in \mc{F}$ is
\bea
F(T,c) := \sum_{e \in T} c_e.
\eea
The regret of a solution $T$ under some cost vector $c$ is the difference between the cost of the solution and the optimal cost:
\bea
R(T,c):= F(T,c) - F^*(c).
\eea 
The \textit{maximum regret problem} for a solution $T$ is 
\bea
R_{\max}(T) := \max_{c \in \mc{C}} R(T,c) = \max_{c \in \mc{C}} \left ( F(T,c) - F^*(c) \right).
\eea
The \textit{deterministic minmax regret problem} is then
\bea
Z_\mathrm{D} := \min_{T \in \mc{F}} R_{\max}(T) = \min_{T \in \mc{F}} \max_{c \in \mc{C}}(F(T,c) - F^*(c)). \label{minmaxd}
\eea
In the remainder of the paper, we will frequently abuse the notation $F(\cdot,c)$, $R(\cdot,c)$ and $R_{\max}(\cdot)$ by replacing set arguments with vectors (e.g. $F(x,c)$ in place of $F(T,c)$), but we will follow the convention of using capital letters for sets and lowercase letters for vectors.

We now move to the randomized framework, where the optimizing player selects a distribution over solutions and the adversary selects a distribution over costs. Starting with the optimizing player, for some set $T \in \mc{F}$, let $y_T$ denote the probability that the optimizing player selects set $T$. Let $y = (y_T)_{T \in \mc{F}}$ be the vector of length $|\mc{F}|$ specifying the set selection distribution; we will refer to $y$ simply as a solution. Define the feasible region for $y$ as 
\bea
\mc{Y} := \{y |y \ge \bs{0}, \bs{1}^\top y = 1\},
\eea
where the notation $\bs{0}$ and $\bs{1}$ indicates a full vector of zeros and ones, respectively. We similarly define a distribution over costs for the adversary. The set $\mc{C}$ may in general be infinite, but we will only consider strategies with finite support; for now we will assume that such strategies are sufficient. Thus consider a finite set $\mc{C}_f \subseteq \mc{C}$, and for some $c \in \mc{C}_f$, let $w_c$ denote the probability that the adversary selects costs $c$. Then let $w = (w_c)_{c \in \mc{C}_f}$ and define the feasible region
\bea
\mc{W} := \{w | w \ge \bs{0}, \bs{1}^\top w = 1\}.
\eea
The expected regret under $y$ and $w$ is simply
\bea
\overline{R}(y,w) := \sum_{T \in \mc{F}} \sum_{c \in \mc{C}_f} y_T w_c  R(T,c) = \sum_{T \in \mc{F}} \sum_{c \in \mc{C}_f} y_T w_c  (F(T,c)-F^*(c)).
\eea
For a given $y$, the \emph{maximum expected regret problem} is
\bea
\overline{R}_{\max}(y) &:=& \max_{w \in \mc{W}} \sum_{c \in \mc{C}_f} w_c \sum_{T \in \mc{F}} y_T R(T,c) \n \\
&=& \max_{c \in \mc{C}_f}\sum_{T \in \mc{F}} y_T R(T,c).
\eea
The above equality follows using the standard observation used in game theory: the optimization of $w \in \mc{W}$ is maximization of the function $G(y,c) = \sum_{T \in \mc{F}} y_T R(T,c)$ over the convex hull of $\mc{C}_f$, which is equivalent to optimizing over $\mc{C}_f$ itself. The minmax expected regret problem, which we refer to as the \textit{randomized minmax regret problem}, is
\bea
Z_\mathrm{R} := \min_{y \in \mc{Y}} \overline{R}_{\max}(y) = \min_{y \in \mc{Y}} \max_{c \in \mc{C}}\left(\sum_{T \in \mc{F}} y_T (F(T,c)-F^*(c)) \right),\label{minmaxer}
\eea
where we have replaced $\mc{C}_f$ with $\mc{C}$ under the assumption that $\mc{C}_f$ contains the maximizing cost vector.

The above minmax expected regret problem is the problem faced by the optimizing player; the adversary, however, is interested in solving the maxmin expected regret problem, defined as follows. First, the \emph{minimum expected regret problem} for a given $w$ is
\bea
\overline{R}_{\min}(w) &:=& \min_{y \in \mc{Y}} \sum_{T \in \mc{F}} y_T \sum_{c \in \mc{C}_f} w_c R(T,c) \n \\
&=& \min_{T \in \mc{F}} \sum_{c \in \mc{C}_f} w_c R(T,c),
\eea
where we have once again used the fact that optimizing over the convex hull of the set of solutions is equivalent to optimizing over the set of solutions. The \emph{adversarial randomized maxmin regret problem} is
\bea
Z_\mathrm{AR} := \max_{w \in \mc{W}} \overline{R}_{\min}(w)= \max_{w \in \mc{W}} \min_{T \in \mc{F}} \left( \sum_{c \in \mc{C}_f} w_c(F(T,c)-F^*(c)) \right).
\eea
It is often the case that the minmax value of the game is equal to the maxmin value; that is, $Z_\mathrm{R} = Z_\mathrm{AR}$. The Minimax Theorem states that this holds for two-person zero-sum games with a finite number of pure strategies \cite{neumann28}. In the following sections, we will show that this identity holds for discrete scenario uncertainty and interval uncertainty, following from linear programming duality.

\section{Discrete Scenario Uncertainty}
\label{discretesec}
Under discrete scenario uncertainty, we are given a finite set $\mc{S}$ of $|\mc{S}| = k$ scenarios. For each $S \in \mc{S}$, there exists a cost vector $c^S = (c_e^S)_{e \in E}$. The adversary's mixed strategy is a probability distribution over scenarios, so we are not concerned with complications arising from infinite sets. This section is divided into three parts; we first determine computation of the optimal randomized strategy for the optimizing player, followed by computation of the adversary's optimal strategy. Thereafter, we use the randomized model to devise a primal-dual approximation scheme for the deterministic minmax regret problem. We restate and clarify some notation in the context of discrete scenario uncertainty throughout our development.

\subsection{Optimizing Player}
We first make some observations regarding the deterministic minmax regret problem that will be helpful in making comparisons with the randomized model. Under discrete scenario uncertainty, the deterministic maximum regret problem is
\bea
R_{\max}(T) = \max_{S \in \mc{S}} R(T,c^S) = \max_{S \in \mc{S}} \left ( F(T,c^S) - F^*(c^S) \right).
\eea
The deterministic minmax regret problem is
\bea
Z_\mathrm{D} = \min_{T \in \mc{F}} R_{\max}(T) = \min_{T \in \mc{F}} \max_{S \in \mc{S}}(F(T,c^S) - F^*(c^S)).
\eea
\begin{lma}
The deterministic minmax regret problem with discrete scenario uncertainty is equivalent to the following integer program.
\begin{alignat}{3}
Z_\mathrm{D} ~=~ \min \quad &z && \label{lp2det}\\
\mathrm{s.t.} \quad & \sum_{e \in E} c^S_e x_e - F^*(c^S) \le z,\qquad  && \forall S \in \mc{S}, \n \\
& \n x \in \mc{X}.
\end{alignat}
\end{lma}
\begin{proof}
Slightly abusing the notation for maximum regret, we have with vector notation
\bea
R_{\max}(x) = \max_{S \in \mc{S}} \left ( \sum_{e \in E} c_e^S x_e - F^*(c^S) \right).
\eea
The integer program then follows by definition of the maximum.
\end{proof}

For the randomized model, recall that the optimizing player's distribution over solutions is denoted by $y = (y_T)_{T \in \mc{F}}$ and that $\mc{Y}$ denotes the set of valid probability distributions. The maximum expected regret problem is
\bea
\overline{R}_{\max}(y) &=& \max_{S \in \mc{S}} \sum_{T \in \mc{F}} y_T R(T,c^S) \n \\
&=& \max_{S \in \mc{S}} \left (\sum_{T \in \mc{F}} y_T F(T,c^S) - F^*(c^S) \right).
\eea
We define the expected value of a solution for a distribution $y$ and cost vector $c^S$ to simplify notation:
\bea
\overline{F}(y,c^S) := \sum_{T \in \mc{F}} y_T F(T,c^S).
\eea
The maximum expected regret problem can then be stated as
\bea
\overline{R}_{\max}(y) = \max_{S \in \mc{S}}(\overline{F}(y,c^S) - F^*(c^S)).
\eea
The randomized minmax regret problem is
\bea
Z_\mathrm{R} = \min_{y \in \mc{Y}} \overline{R}_{\max}(y) = \min_{y \in \mc{Y}} \max_{S \in \mc{S}}(\overline{F}(y,c^S) - F^*(c^S)).
\eea

To solve the randomized minmax regret problem, it is possible to write a linear program analogous to the above integer program using variables $y_T$. This would, however, have $|\mc{F}|$ variables, which may grow exponentially in $n$. Instead, we note that for the maximum regret expected regret problem,
\begin{flalign}
\overline{R}_{\max}(y)&=\max_{S \in \mc{S}} \left (\sum_{T\in \mc{F}}y_T \sum_{e \in T} c_e^S - F^*(c^S) \right)\n \\
&=\max_{S \in \mc{S}} \left(\sum_{e \in E} c_e^S \sum_{T \in \mc{F}: e \in T} y_T - F^*(c^S) \right).
\end{flalign}
The change in summation order motivates the substitution
\bea
p_e := \sum_{T \in \mc{F} : e \in T} y_T, \quad e \in E. \label{subpe}
\eea
Let $p = (p_1,\ldots,p_n)$; we will refer to this as the marginal probability vector. The substitution is a mapping from $\mc{Y}$ to the convex hull of $\mc{X}$. The following is the minmax regret analog of an observation made by Bertsimas et al. \cite{bert12}.

\begin{lma}
For discrete scenario uncertainty, the objective value $Z_{\mathrm{R}}$ of the randomized minmax regret problem \eqref{minmaxer} is equal to that of the problem 
\bea
\min_{p \in \ch{(\mc{X})}} \max_{S \in \mc{S}}  \left(\sum_{e \in E} c_e^S p_e - F^*(c^S) \right),
\eea
where $\ch{(\mc{X})}$ denotes the convex hull of $\mc{X}$.
\label{pch}
\end{lma}
\begin{proof}
We use the same arguments presented in \cite{bert12}. By definition of the substitution \eqref{subpe}, the vector $p$ must lie in the convex hull of $\mc{X}$. Carath\'eodory's Theorem  \cite{cara11} states that any $p \in \ch(\mc{X})$ can be represented by a convex combination of at most $n+1$ points in $\mc{X}$, so there exists a surjective mapping from $\mc{Y}$ to $\ch(\mc{X})$.
\end{proof}

Since we will use the simplified formulation given in Lemma \ref{pch} to solve the randomized minmax regret problem, we address the problem of recovering a vector $y$ given a solution $p$. In the proof of the lemma, we have used Carath\'eodory's Theorem, which proves existence of such a mapping, but not its construction. To this end, we define for the optimizing player a \textit{mixed strategy encoding} $\mc{M} = (X,Y)$ as a set of deterministic solutions $X = \{x^{T_i} \in \mc{X},~i = 1,\ldots, \mu \}$ that should be selected with nonzero probability and the corresponding probabilities $Y = \{y_{T_i} \in [0,1],~i = 1,\ldots, \mu \}$ that satisfy $\sum_{i=1}^\mu y_{T_i} = 1$. Here $\mu$ is the support size of the mixed strategy (i.e. the number of deterministic solutions with nonzero probablity). For a given vector $p$, we are interested in solving the following constraint satisfaction program:
\begin{alignat}{3}
\min \quad &0 && \label{constrsat}\\
\mathrm{s.t.} \quad & \sum_{T \in \mc{F}:e \in T} y_T = p_e,\qquad  && \forall e \in E, \n \\
& \sum_{T \in \mc{F}} y_T = 1, \n \\
& y \ge \bs{0}. \n
\end{alignat}
Consider the dual program of \eqref{constrsat}, which has variables $u = (u_1,\ldots,u_e)$ and $w$:
\begin{alignat}{3}
\max \quad &w-\sum_{e \in E} p_e u_e && \label{constrsatdual} \\
\mathrm{s.t.} \quad &w- \sum_{e \in T} u_e \le 0,\qquad  && \forall T \in \mc{F}, \label{cc} \\
& u,w\mathrm{~free}. \n
\end{alignat}
Recall that the region $\mc{X}$ is described in size $m$.
\begin{lma}
\label{recover}
For any given $p \in \ch{(\mc{X})}$, a corresponding \textit{mixed strategy encoding} $\mc{M}$ of size polynomial in $n$ can be found via the linear programming formulation \eqref{constrsatdual} - \eqref{cc}. Furthermore, if the nominal problem $F^*(c)$ can be solved in time polynomial in $n$ and $m$, then $\mc{M}$ can be found in time polynomial in $n$ and $m$.
\end{lma}
\begin{proof}

Notice that while the primal program has an exponential number of variables and a linear number of constraints, the opposite holds true for the dual. The primal program is bounded since all objective coefficients are equal to zero, and is feasible due to Carath\'{e}odory's Theorem. Therefore the dual program must be feasible and bounded.

To guarantee a polynomial sized solution, note that the separation problem for the constraints \eqref{cc} is simply the nominal problem with costs $u$, so the dual program can be solved via the ellipsoid method. If the nominal problem can be solved in polynomial time, then the constraints \eqref{cc} can be generated in polynomial time, giving a polynomial time solution for the entire dual program \eqref{constrsatdual} - \eqref{cc}.

From a practical perspective, a separation oracle for \eqref{cc} gives an efficient method for performing row generation with the simplex method. Each row $i$ generated while solving the dual problem gives a solution $x^{T_i} \in \mc{X}$, and its dual variable is the corresponding probability $y_{T_i}$.
\end{proof}

Using Lemma \ref{pch}, we can now formulate a linear program to solve the randomized minmax regret problem.
\begin{alignat}{3}
\min \quad &z && \label{lp2}\\
\mathrm{s.t.} \quad & \sum_{e \in E} c_e^S p_e - F^*(c^S) \le z,\qquad  && \forall S \in \mc{S}, \label{lconstr} \\
& p \in \ch{(\mc{X})} \label{lp2end}. 
\end{alignat}
This leads to the important result that the randomized minmax regret problem is polynomial solvable for any polynomial solvable nominal problem. Also, the minmax expected regret is upper bounded by the minmax regret in the deterministic case.
\begin{thm}
For discrete scenario uncertainty, if the nominal problem $F^*(c)$ can be solved in time polynomial in $n$ and $m$, then the corresponding randomized minmax regret problem \\ $\min_{y \in \mc{Y}} \max_{S \in \mc{S}}(\overline{F}(y,c^S) - F^*(c^S))$ can be solved in time polynomial in $n$, $m$, and $k$.
\end{thm}
\begin{proof}
Since for all $S \in \mc{S}$, the value $F^*(c^S)$ is polynomial solvable, each constraint \eqref{lconstr} can be enumerated in polynomial time. 
If we can optimize over $\mc{X}$ in polynomial time, then we can separate over $\ch{(\mc{X})}$ in polynomial time via the result of \cite{grot81}. This gives the separation oracle for \eqref{lp2end}.
\end{proof}
\begin{coro} 
For discrete scenario uncertainty, $Z_\mathrm{R} \le Z_\mathrm{D}.$
\end{coro}
\begin{proof}
The program \eqref{lp2} - \eqref{lp2end} is the linear programming relaxation of \eqref{lp2det}.
\end{proof}

\subsection{Adversary}
Moving to the perspective of the adversary under discrete scenario uncertainty, the adversary must select a mixed strategy over scenarios. The finite number of scenarios naturally requires the adversary's distribution to have finite support.  Specifically, the adversary selects a distribution over costs $w = (w_S)_{S \in \mc{S}}$. The minimum expected regret problem for a given $w$ is
\bea
\overline{R}_{\min}(w) &=& \min_{T \in \mc{F}} \sum_{S \in \mc{S}} w_S R(T,c^S) \n \\
&=& \min_{T \in \mc{F}} \sum_{S \in \mc{S}} w_S\left(F(T,c^S)-F^*(c^S)\right).
\eea
Recall that $\mc{W}$ indicates valid probability distributions for $w$. The adversarial randomized maxmin regret problem is
\bea
Z_{\mathrm{AR}} = \max_{w \in \mc{W}} \overline{R}_{\min}(w) = \max_{w \in \mc{W}} \min_{T \in \mc{F}}  \sum_{S \in \mc{S}} w_S\left(F(T,c^S)-F^*(c^S)\right).
\eea

From the above definition, we formulate a linear program to solve the adversarial randomized maxmin regret problem:
\begin{alignat}{3}
\max \quad &z && \label{adv_lp_d1}\\
\mathrm{s.t.} \quad & \sum_{S \in \mc{S}} w_S (F(T,c^S) - F^*(c^S)) \ge z,\qquad  && \forall T \in \mc{F}, \label{advx} \\
& w \in \mc{W} \label{adv_lp_d3}.
\end{alignat}
The linear program has an exponential number of constraints, but the nominal problem gives a separation oracle.
\begin{thm}
For discrete scenario uncertainty, if the nominal problem $F^*(c)$ can be solved in time polynomial in $n$ and $m$, then the corresponding randomized adversarial maxmin regret problem \\ $\max_{w \in \mc{W}} \min_{T \in \mc{F}}  \sum_{S \in \mc{S}} w_S\left(F(T,c^S)-F^*(c^S)\right)$ can be solved in time polynomial in $n$, $m$, and $k$.
\end{thm}
\begin{proof}
The separation oracle for \eqref{advx} is given by the nominal problem. First, notice that $F^*(c^S)$ for $S \in \mc{S}$ can be computed once at initialization and then stored for easy computation of $\sum_{S \in \mc{S}} w_S F^*(c^S)$ for any $w$. Next, we have
\bea
\sum_{S \in \mc{S}} w_S F(T,c^S) = \sum_{S \in \mc{S}} w_S \sum_{e \in T} c_e^S  = \sum_{e \in T} \left( \sum_{S \in \mc{S}} w_S c_e^S \right).
\eea
This means that solving nominal problem with costs $d = (d_1,\ldots,d_n)$ where
\bea
d_e = \sum_{S \in \mc{S}} w_S c_e^S
\eea
and comparing the solution with $z$ and $\sum_{S \in \mc{S}} w_S F^*(c^S)$ gives the oracle.
\end{proof}

\begin{coro}
For discrete scenario uncertainty, $Z_{\mathrm{R}} = Z_{\mathrm{AR}}$.
\end{coro}
\begin{proof}
Using the substitution of the marginal probability vector in \eqref{subpe}, it can be verified that the linear program solved by the adversary \eqref{adv_lp_d1} - \eqref{adv_lp_d3} is the dual of the program solved by the optimizing player \eqref{lp2} - \eqref{lp2end}. The result holds by strong duality.
\end{proof}

\subsection{Primal-Dual Approximation}
As noted in the above corollary, the linear program solved by the adversary \eqref{adv_lp_d1} - \eqref{adv_lp_d3}  is the dual of program solved by the optimizing player \eqref{lp2} - \eqref{lp2end}. These linear programs correspond to the relaxation of the deterministic minmax regret problem, and can thus be used to develop a primal-dual approximation scheme. We will refer to the program solved by the optimizing player as the primal linear program, and the problem solved by the adversary as the dual linear program.

We rewrite the dual program \eqref{adv_lp_d1} - \eqref{adv_lp_d3} as
\begin{alignat}{3}
\max \quad &z - \sum_{S \in \mc{S}} w_S F^*(c^S)&&\\
\mathrm{s.t.} \quad & \sum_{S \in \mc{S}} w_S F(T,c^S)\ge z,\qquad  && \forall T \in \mc{F}, \label{wft} \\
& w \in \mc{W}.
\end{alignat}
A simple feasible solution to this program is given first by setting $w_S = 1 / k$ for each $S \in \mc{S}$. Using the standard approach for primal-dual algorithms \cite{williamson11}, we start with a sufficiently small value of $z$ and increase it until a constraint becomes tight. The set corresponding to the tight solution is then added to the primal solution. The constraint \eqref{wft} can be written as
\bea
\sum_{S \in \mc{S}} w_S F(T,c^S) = \sum_{S \in \mc{S}} w_S \sum_{e \in T} c_e^S = \sum_{e \in T}\left( \frac{1}{k} \sum_{S \in \mc{S}} c_e^S \right).
\eea
The first constraint that becomes tight corresponds to the set $M$ that minimizes the mean costs over all scenarios,
\bea
M := \argmin_{T \in \mc{F}} \sum_{e \in T}\left( \frac{1}{k} \sum_{S \in \mc{S}} c_e^S \right).
\eea
The set $M$, which is complete primal feasible solution, is added to the primal problem. Additionally, we have a feasible solution to the adversarial (dual) linear program with objective value
\bea
\left(\frac{1}{k}\right) \sum_{S \in \mc{S}} \left( \sum_{e \in M} c_e^S - F^*(c^S) \right),
\eea
which is a lower bound for the optimal objective value $Z_{\mathrm{R}}$. Using the same observations made in \cite{aissi06}, this gives a $k$-approximation algorithm for the minmax regret problem. The result given by the primal-dual framework is stronger than the result proved in \cite{aissi06} since it bounds the value of the approximate solution within a factor $k$ of the linear programming relaxation value $Z_{\mathrm{R}}$, rather than the integer program value $Z_{\mathrm{D}}$.
\begin{thm}
For discrete scenario uncertainty, the solution to the nominal problem with mean costs is a $k$-approximation algorithm for the deterministic minmax regret problem.
\end{thm}
\begin{proof}
Using the construction above for a lower bound on $Z_{\mathrm{R}}$, we have
\bea
\frac{Z_{\mathrm{D}}}{k} \le \left( \frac{1}{k} \right) \max_{S \in \mc{S}} \left( \sum_{e \in M} c_e^S - F^*(c^S) \right) \le \left(\frac{1}{k}\right) \sum_{S \in \mc{S}} \left( \sum_{e \in M} c_e^S - F^*(c^S) \right) \le Z_{\mathrm{R}}.
\eea
The first inequality follows by definition of the deterministic minmax regret, the second inequality by a simple identity between the sum of a set of values and the maximum, and the third inequality from the linear program.
\end{proof}
An interesting corollary is a tight bound on the power of randomization in the minmax regret problem. For any nominal problem, moving from a deterministic solution to a randomized solution allows the optimizing player to at most reduce the expected regret by a factor of $k$.
\begin{coro}
For discrete scenario uncertainty,
\bea
Z_{\mathrm{R}} \ge \frac{Z_{\mathrm{D}}}{k}.
\eea
\label{igk}
\end{coro}
\nin The corollary equivalently states that the \emph{integrality gap}, defined as the largest possible ratio of the optimal objective value of a program to its optimal linear programming relaxation, is equal to $k$. This holds independent of the nominal problem. We construct a tight example for the corollary using $n = k$ items, where the goal of the problem is simply to select the single item with lowest cost. For each item, there exists a scenario where the item has cost $c_e = 1$ and all other items have costs $c_e = 0$. The deterministic minmax regret is equal to $1$ for the problem. In the randomized problem, the optimizing player selects each item with probability $1/k$ and the adversary selects each scenario with probability $1/k$. The expected regret is equal to the probability that the optimizing player selects the same item that the adversary assigns unit cost to, which is equal to $1/k$.

\section{Interval Uncertainty}
\label{intervalsec}
In this section we assume that cost uncertainty is characterized by interval uncertainty, meaning that each item cost is independently contained within known lower and upper bounds:
\bea
c_e \in [c_e^-, c_e^+],\quad \forall e \in E.
\eea
Define the region
\bea
\mc{I} := \{c | c_e \in [c_e^-, c_e^+], e \in E\}.
\eea
The set $\mc{I}$ is in general infinite. Since we wish to use a mixed distribution over $\mc{I}$ with finite support, we loosely define the set $\mc{I}_f $ to be some subset $\mc{I}_f \subset \mc{I}$ with finite cardinality, over which a probability distribution will be defined. The exact construction of $\mc{I}_f$ will become clear during the analysis, but a sufficient example is the set of cost vectors where costs are set equal to their lower or upper bounds, $\mc{I}_f = \{c | c_e = c_e^-~\mathrm{or}~c_e = c_e^+ , e \in E\}$.

We proceed in the same way as the last section, studying the optimal policy for the optimizing player and then the adversary, followed by a primal-dual approximation algorithm for the deterministic problem. We restate notation and definitions throughout. 

\subsection{Optimizing Player}
Under interval uncertainty, we have the deterministic maximum regret problem
\bea
R_{\max}(T) = \max_{c \in \mc{I}} R(T,c) = \max_{c \in \mc{I}} \left ( F(T,c) - F^*(c) \right) \label{maxregd}
\eea
and the deterministic minmax regret problem
\bea
Z_\mathrm{D} = \min_{T \in \mc{F}} R_{\max}(T) = \min_{T \in \mc{F}} \max_{c \in \mc{I}}(F(T,c) - F^*(c)).
\eea
The deterministic minmax regret problem is well studied and can be solved with a mixed integer program \cite{kasperski08}. We use an unconventional formulation, which has an exponential number of constraints. We will ultimately show that the randomized minmax regret problem corresponds to the linear programming relaxation of this formulation.
\begin{lma}
For interval uncertainty, the deterministic minmax regret problem \eqref{minmaxd} is equivalent to the following integer program.
\begin{alignat}{3}
Z_\mathrm{D} ~=~ \min \quad &z && \label{lp1det}\\
\mathrm{s.t.} \quad & \sum_{e \in E \setminus T} c_e^+x_e-\sum_{e \in T} c_e^- \left( 1- x_e\right) \le z,\qquad  && \forall T \in \mc{F}, \n \\
& \n x \in \mc{X}.
\end{alignat}
\label{detexp}
\end{lma}
\begin{proof}
From the maximum regret definition \eqref{maxregd} and using vector notation instead of set notation,
\bea
R_{\max}(x) &=& \max_{c \in \mc{I}} \left ( F(x,c) - F^*(c) \right) \n \\
&=& \max_{c \in \mc{I}} \left ( \sum_{e \in E} c_e x_e - \min_{T \in \mc{F}} \sum_{e \in T} c_e \right) \n \\
&=& \max_{T \in \mc{F}} \max_{c \in \mc{I}} \left( \sum_{e \in E} c_e x_e - \sum_{e \in T} c_e \right) \n \\
&=& \max_{T \in \mc{F}} \max_{c \in \mc{I}} \left( \sum_{e \in E \setminus T} c_e x_e - \sum_{e \in T} c_e (1-x_e) \right) \n \\
&=& \max_{T \in \mc{F}} \left( \sum_{e \in E \setminus T} c_e^+ x_e - \sum_{e \in T} c_e^- (1-x_e) \right),
\eea
where in the third equality we have used that the expression $\sum_{e \in E} c_e x_e$ is not a function of $T$, and the last equality follows since $x_e \in \{0,1\}$. 
The program is then valid by the definition of the maximum.
\end{proof}

In the randomized model, the maximum expected regret is
\bea
\overline{R}_{\max}(y) &=& \max_{c \in \mc{I}} \sum_{T \in \mc{F}} y_T R(T,c) \n \\
&=& \max_{c \in \mc{I}} \left (\sum_{T \in \mc{F}} y_T F(T,c) - F^*(c) \right) \label{maxer}.
\eea
As with the discrete scenario uncertainty case, we define the expected value of a solution for a distribution $y$ and cost vector $c$,
\bea
\overline{F}(y,c) := \sum_{T \in \mc{F}} y_T F(T,c),
\eea
so the maximum expected regret can be stated as
\bea
\overline{R}_{\max}(y) = \max_{c \in \mc{I}}(\overline{F}(y,c) - F^*(c)).
\eea
The randomized minmax regret problem is thus
\bea
Z_\mathrm{R} = \min_{y \in \mc{Y}} \overline{R}_{\max}(y) = \min_{y \in \mc{Y}} \max_{c \in \mc{I}}(\overline{F}(y,c) - F^*(c)).
\eea

Starting with analysis of the maximum expected regret problem \eqref{maxer}, we use the same substitution that we used in the previous section. Specifically, we let
\bea
p_e = \sum_{U \in \mc{F}:e \in U} y_U, \quad e \in E,
\eea
and define the marginal probability vector $p = (p_1,\ldots,p_n)$. Slightly abusing notation, we write $\overline{R}_{\max}(p)$ in place of $\overline{R}_{\max}(y)$ via this substitution.
\begin{lma}
\label{lama}
For interval uncertainty, the maximum expected regret problem \eqref{maxer} is equivalent to the problem
\bea
\overline{R}_{\max}(p) = \max_{T \in \mc{F}} \left( \sum_{e \in E \setminus T} c_e^+ p_e -\sum_{e \in T} c_e^- \left( 1- p_e\right) \right).
\eea
\end{lma}
\begin{proof} We start with \eqref{maxer} and use the substitution of $p$. The analysis is nearly identical to the proof of Lemma \ref{detexp}.
\begin{flalign}
\overline{R}_{\max}(y)&=  \max_{c \in \mc{I}}(\overline{F}(y,c) - F^*(c)) \n \\
&=\max_{c \in \mc{I}} \left (\sum_{U\in \mc{F}}y_U \sum_{e \in U} c_e - \min_{T \in \mc{F}} \left ( \sum_{e \in T} c_e \right) \right)\n \\
&=\max_{c \in \mc{I}} \left(\sum_{e \in E} c_e \sum_{U \in \mc{F}: e \in U} y_U - \min_{T \in \mc{F}} \left ( \sum_{e \in T} c_e \right) \right)\n \\
&=\max_{c \in \mc{I}} \left(\sum_{e \in E} c_e p_e - \min_{T \in \mc{F}} \left ( \sum_{e \in T} c_e \right) \right).
\end{flalign}
Now using the notation $\overline{R}_{\max}(p)$,
\begin{flalign}
\overline{R}_{\max}(p)&=  \max_{c \in \mc{I}} \left(\sum_{e \in E} c_e p_e - \min_{T \in \mc{F}} \left ( \sum_{e \in T} c_e \right) \right) \n \\
&=\max_{c \in \mc{I}} \max_{T \in \mc{F}} \left(\sum_{e \in E} c_e p_e -\sum_{e \in T} c_e  \right)\n \\
&= \max_{T \in \mc{F}} \max_{c \in \mc{I}} \left(\sum_{e \in E \setminus T} c_e p_e -\sum_{e \in T} c_e \left( 1- p_e\right)  \right),\label{pause}
\end{flalign}
where in the first equality we have used that the expression $\sum_{e \in E} c_e p_e$ is not a function of $T$, and the other equalities follow from rearranging terms. Notice in \eqref{pause} that $p_e$ is simply the total probability that item $e$ is selected, so for $y \in \mc{Y}$, we must have $p_e \in [0,1]$. This makes it easy to see that for a given $T \in \mc{F}$,
\bea
\max_{c \in \mc{I}} \left(\sum_{e \in E \setminus T} c_e p_e -\sum_{e \in T} c_e \left( 1- p_e\right)  \right) =\sum_{e \in E \setminus T} c_e^+ p_e -\sum_{e \in T} c_e^- \left( 1- p_e\right)  \label{pause2}.
\eea
Substituting \eqref{pause2} into \eqref{pause} then gives an optimization problem with a finite number of feasible solutions,
\be
\overline{R}_{\max}(p) =\max_{T \in \mc{F}} \left ( \sum_{e \in E \setminus T} c_e^+ p_e  -\sum_{e \in T} c_e^- \left( 1- p_e \right) \right),
\label{endproof}
\ee
which completes the proof.
\end{proof}

An immediate corollary of Lemma \ref{lama} is that we can solve the maximum expected regret problem for a given $y$ by enumerating all $|\mc{F}|$ subsets (potentially an exponential number of them) and choosing the one that maximizes the argument of \eqref{endproof}. This allows the entire randomized minmax regret problem to be restated.
\begin{lma}
For interval uncertainty, the objective value $Z_\mathrm{R}$ of the randomized minmax regret problem \eqref{minmaxer} is equal to that of the problem
\bea
\min_{p \in \ch(\mc{X})} \max_{T \in \mc{F}} \left ( \sum_{e \in E \setminus T} c_e^+p_e-\sum_{e \in T} c_e^- \left( 1- p_e\right) \right), \label{motivlemma}
\eea
where $\ch(\mc{X})$ denotes the convex hull of $\mc{X}$.
\label{pch2}
\end{lma}
\begin{proof}
By the same argument as the proof of Lemma \ref{pch}.
\end{proof}
Using Lemma \ref{pch2}, we can now formulate a linear program to solve the randomized minmax regret problem:
\begin{alignat}{3}
~\min \quad &z && \label{lp1}\\
\mathrm{s.t.} \quad & \sum_{e \in E \setminus T} c_e^+p_e-\sum_{e \in T} c_e^- \left( 1- p_e\right) \le z,\qquad  && \forall T \in \mc{F}, \label{sep1}\\
& p \in \ch{(\mc{X})}. \label{sep2}
\end{alignat}
\nin While the above program may have an exponential number of constraints, it can be solved efficiently via the ellipsoid algorithm if a separation oracle is available for the constraints \eqref{sep1} and \eqref{sep2}. This brings us to our main result.
\begin{thm}
For interval uncertainty, if the nominal problem $F^*(c)$ can be solved in time polynomial in $n$ and $m$, then the corresponding randomized minmax regret problem $\min_{y \in \mc{Y}} \max_{c \in \mc{I}}(\overline{F}(y,c) - F^*(c))$ can be solved in time polynomial in $n$ and $m$.
\end{thm}
\begin{proof}
Consider the linear program \eqref{lp1} - \eqref{sep2}. The separation oracle for the constraints \eqref{sep2} is given by the equivalence of optimization and separation \cite{grot81}. To see the separation oracle for the constraint \eqref{sep1}, we define the item cost vector $d = (d_1,\ldots,d_n)$ where
\bea
d_e = c_e^- + p_e(c_e^+-c_e^-), \quad e \in E,
\eea
and then solve
\bea
z_d = \min_{T \in \mc{F}} \sum_{e \in T} d_e.
\eea
Let $T_d$ be the set that minimizes the above expression. If $\sum_{e \in E} c_e^+ p_e - z_d \le z$, then we are guaranteed feasibility, otherwise the separating hyperplane \eqref{sep1} is generated where $T = T_d$. To see the validity of this approach, we have
\bea
\sum_{e \in E} c_e^+ p_e - z_d &=&\sum_{e \in E} c_e^+ p_e - \min_{T \in \mc{F}} \sum_{e \in T} d_e \n \\
&=&\sum_{e \in E} c_e^+ p_e - \min_{T\in \mc{F}} \sum_{e \in T} (c_e^- + p_e(c_e^+-c_e^-))\n \\
&=&\max_{T \in \mc{F}} \left (\sum_{e \in E} c_e^+ p_e -\sum_{e \in T} (c_e^- + p_e(c_e^+-c_e^-)) \right)\n \\
&=&\max_{T \in \mc{F}} \left (\sum_{e \in E \setminus T} c_e^+ p_e -\sum_{e \in T} c_e^-(1-p_e) \right).
\eea
The solution to the linear program \eqref{lp1}~-~\eqref{sep2} is a vector $p$, which can then be used to find a mixed strategy $y$ in polynomial time using Lemma \ref{recover}.
\end{proof}
\begin{coro} 
For interval uncertainty, $Z_\mathrm{R} \le Z_\mathrm{D}.$
\label{zrdinterval}
\end{coro}
\begin{proof}
This follows simply by noting that the program \eqref{lp1} - \eqref{sep2} is the linear programming relaxation of \eqref{lp1det}.
\end{proof}

\subsection{Adversary}
The set $\mc{I}_f$ is necessary for describing the distribution of the adversary. The distribution over costs is $w = (w_c)_{c \in \mc{I}_f}$ and $\mc{W}$ again indicates the set of valid distributions. The minimum expected regret problem is
\bea
\overline{R}_{\min}(w) &=& \min_{T \in \mc{F}} \sum_{c \in \mc{I}_f} w_c R(T,c) \n \\
&=& \min_{T \in \mc{F}} \sum_{c \in \mc{I}_f} w_c\left(F(T,c)-F^*(c)\right).
\eea
The adversarial randomized maxmin regret problem is then
\bea
Z_{\mathrm{AR}} = \max_{w \in \mc{W}} \overline{R}_{\min}(w) = \max_{w \in \mc{W}} \min_{T \in \mc{F}}  \sum_{c \in \mc{I}_f} w_c\left(F(T,c)-F^*(c)\right).
\eea

We can directly formulate a linear program for the adversarial maxmin regret problem, explicitly writing the constraints for $w \in \mc{W}$.
\begin{alignat}{4}
\max \quad &z &&& \label{advstart1} \\
\mathrm{s.t.} \quad & \sum_{c \in \mc{I}_f} w_c \left(\sum_{e \in T} c_e-F^*(c) \right) \ge z,\qquad  && \forall T \in \mc{F},\qquad&\\
& \sum_{c \in \mc{I}_f} w_c = 1, &&& \\
& w_c \ge \bs{0}. \label{advstart2} &&&
\end{alignat}
Since this program has an exponential number of constraints and potentially an exponential number of variables, we consider its dual. We expect this ``dual of the dual" program to be the primal linear program solved by the optimizing player; this will indeed be the case after some manipulation. The dual of \eqref{advstart1} - \eqref{advstart2} is 
\begin{alignat}{4}
\min \quad &\beta &&& \\
\mathrm{s.t.} \quad & \sum_{T \in \mc{F}} \alpha_T \left(\sum_{e \in T} c_e-F^*(c) \right) \le \beta,\qquad  && \forall c \in \mc{I}_f,\qquad& \\
& \sum_{T\in \mc{F}} \alpha_T = 1, &&& \\
& \alpha_T\ge \bs{0}. &&&
\end{alignat}
To simplify, note that for a feasible $\alpha = (\alpha_T)_{T \in \mc{F}}$, we have $\sum_{T \in \mc{F}} \alpha_T F^*(c) = F^*(c)$. Furthermore,
\bea
\sum_{T \in \mc{F}} \alpha_T \sum_{e \in T} c_e = \sum_{e \in E} c_e \sum_{T \in \mc{F}: e \in T} \alpha_T.
\eea
We use the substitution
\bea
q_e := \sum_{T \in \mc{F}:e \in T} \alpha_T, \quad e \in E.
\eea
Let $q = (q_1,\ldots,q_n)$. The substitution yields the linear program
\begin{alignat}{4}
\min \quad &\beta &&& \\
\mathrm{s.t.} \quad & \sum_{e \in E} c_e q_e - F^*(c) \le \beta,\qquad  && \forall c \in \mc{I}_f,\qquad&\label{aconstrq}\\
& q \in \ch{(\mc{X})}. &&&
\end{alignat}
This program no longer has an exponential number of variables, and the exponential number of constraints can be handled via separation, which we describe shortly. First, for some set $A \in \mc{F}$, define the cost vector $c^A = (c_e^A)_{e \in E}$ where
\bea
\label{cadef}
c_e^A := \left \{
\begin{array}{ll}
c_e^-,& e \in A, \\
c_e^+,& e \in E\setminus A.
\end{array}
\right. 
\eea
That is, $c^A$ is the cost vector where all costs are equal to their upper bound, except for costs in the set $A$, which are equal to their lower bound. The theorem below shows that without loss of generality, we can can consider cost vectors of this form for separation of the constraint \eqref{aconstrq}. This allows us to define $\mc{I}_f$ as the set of all cost vectors $\{c^A, ~A \in \mc{F} \}$. Since this may still be an exponentially sized set, we define an adversarial mixed strategy encoding $\mc{L} = (C,W)$ as a set of costs $C = \{c^{A_j} \in \mc{I}_f,~ j = 1,\ldots,\eta\}$ to be selected with corresponding probabilities $W = \{w_{c^{A_j}} \in [0,1],~j=1,\ldots,\eta\}$ satisfying $\sum_{j = 1}^\eta w_{c^{A_j}} = 1$.
\begin{thm}
For interval uncertainty, if the nominal problem $F^*(c)$ can be solved in time polynomial in $n$ and $m$, then the corresponding randomized adversarial maxmin regret problem \\$\max_{w \in \mc{W}} \min_{T \in \mc{F}}  \sum_{c \in \mc{I}_f} w_S\left(F(T,c)-F^*(c)\right)$ can be solved in time polynomial in $n$ and $m$.
\end{thm}
\begin{proof}
The constraint \eqref{aconstrq} is simply the maximum regret problem for a given vector $q$ and can be generated via the nominal problem,
\bea
\max_{c \in \mc{I}_f} \left( \sum_{e \in E}  c_e q_e - F^*(c) \right) &=& \max_{T \in \mc{F}} \left (\sum_{e \in E \setminus T} c_e^+ q_e -\sum_{e \in T} c_e^-(1-q_e) \right),
\eea
where we have used the analysis in Lemma \ref{lama}. This allows us to write the linear program as
\begin{alignat}{4}
\min \quad &\beta \label{advrepfirst}&&& \\
\mathrm{s.t.} \quad & \sum_{e \in E \setminus T} c_e^+ q_e -\sum_{e \in T} c_e^-(1-q_e) \le \beta,\qquad  && \forall T \in \mc{F},\qquad& \label{advrepmiddle}\\
& q \in \ch{(\mc{X})} \label{advreplast}. &&&
\end{alignat}
Note this is precisely the linear program \eqref{lp1} - \eqref{sep2} solved by the optimizing player.
This justifies the assumption of the finite set $\mc{I}_f$: only a polynomial number of separating cost vectors will be generated, and they will be of the form $c^A$ as defined in \eqref{cadef}. The adversary is of course interested in the dual variables of the linear program \eqref{advrepfirst} - \eqref{advreplast}. Each separating hyperplane generated for the constraint \eqref{advrepmiddle} gives a set $T \in \mc{F}$ for which the adversary adds the cost vector $c^T$ to his mixed strategy; this cost vector has probability $w_{c^T}$ in the mixed strategy, given by the corresponding dual variable.
\end{proof}
\begin{coro}
For interval uncertainty, $Z_{\mathrm{R}} = Z_{\mathrm{AR}}$.
\end{coro}
\begin{proof}
By strong duality, since the optimizing player solves \eqref{lp1}-\eqref{sep2}, which is the dual of the linear program \eqref{advstart1} - \eqref{advstart2} for the adversary.
\end{proof}

\subsection{Primal-Dual Approximation}
The primal linear program for the optimizing player is \eqref{lp1}-\eqref{sep2} and the dual linear program for the adversary is \eqref{advstart1} - \eqref{advstart2}. Using a similar approach to the previous section, we devise a primal-dual approximation algorithm for the deterministic minmax regret problem.

We rewrite the dual linear program as
\begin{alignat}{4}
\max \quad &z - \sum_{c \in \mc{I}_f} w_c F^*(c) &&& \\
\mathrm{s.t.} \quad & \sum_{c \in \mc{I}_f} w_c \sum_{e \in U} c_e \ge z,\qquad  && \forall U \in \mc{F},\qquad&\\
& \sum_{c \in \mc{I}_f} w_c = 1, &&&\\
& w_c \ge \bs{0}. &&&
\end{alignat}
We must select a feasible solution for $w$; we will make the simple choice of setting $w_c = 1/2$ for two cost vectors. 
Recall the definition of the cost vector $c^A$. For some set $A \in \mc{F}$, we have $c^A = (c^A_e)_{e \in E}$ where
\bea
c_e^A = \left \{
\begin{array}{ll}
c_e^-,& e \in A, \\
c_e^+,& e \in E\setminus A,
\end{array}
\right. 
\eea
Additionally, we define $c^{\overline{A}} = (c^{\overline{A}}_e)_{e \in E}$ where
\bea
c_e^{\overline{A}} = \left \{
\begin{array}{ll}
c_e^+,& e \in A, \\
c_e^-,& e \in E\setminus A.
\end{array}
\right.
\eea
We set $w_c = 1/2$ for $c = c^A$ and $c = c^{\overline{A}}$, so the linear program becomes
\begin{alignat}{4}
\max \quad &z - \frac{1}{2} \left (F^*(c^A) - F^*(c^{\overline{A}}) \right) &&& \\
\mathrm{s.t.} \quad & \sum_{e \in U} \left( \frac{c_e^-+c_e^+}{2} \right) \ge z,\qquad  && \forall U \in \mc{F}.\qquad&
\end{alignat}
Under the primal dual approach, we increase $z$ until one of the constraints becomes tight. The first tight constraint corresponds to the primal solution $M$ that has minimum total cost under the midpoint costs:
\bea
M := \argmin_{U \in \mc{F}} \sum_{e \in U} \left( \frac{c_e^- + c_e^+}{2} \right).
\eea 
This gives the objective value
\bea
\frac{1}{2} \left( \sum_{e \in M} ( c_e^- + c_e^+) - F^*(c^A) - F^*(c^{\overline{A}}) \right).
\eea
By choosing the set $A$ to be the midpoint cost minimizing set $M$, we can write the resulting objective value in terms of the maximum deterministic regret.
\begin{lma}
For $A = M$,
\bea
\sum_{e \in M} ( c_e^- + c_e^+) - F^*(c^A) - F^*(c^{\overline{A}}) = R_{\max}(M).
\eea
\label{lmaam}
\end{lma}
\begin{proof}
The maximum regret for the set $M$ can be expressed as
\bea
R_{\max}(M) = \sum_{e \in M} c_e^+ - \min_{T \in \mc{F}}  \left ( \sum_{e \in T \cap M} c_e^+ + \sum_{e \in T \setminus M} c_e^- \right).
\eea
Also note that
\bea
F^*(c^{\overline{A}}) = \min_{T \in \mc{F}} \left ( \sum_{e \in T \cap A} c_e^+ + \sum_{e \in T \setminus A} c_e^- \right).
\eea
Thus for $A = M$, we have
\bea
R_{\max}(M) = \sum_{e \in M} c_e^+ - F^*(c^{\overline{A}}).
\eea
It is left to show that $F^*(c^M) = \sum_{e \in M} c_e^-$. This, however, immediately follows with a simple argument. Since the set $M$ is minimum for midpoint costs, it must also be minimum for costs $c^M$ (i.e., the costs where $c_e = c_e^-$ for all $e \in M$ and $c_e = c_e^+$ for all $e \in E \setminus M$).
\end{proof}
This gives a new proof that solving the nominal problem with midpoint costs gives a $2$-approximation to the deterministic minmax regret problem. Once again, the result here is stronger than the result of \cite{kasperski06} since it states that the value of the approximate solution is within a factor $2$ of $Z_{\mathrm{R}}$ rather than just $Z_{\mathrm{D}}$.
\begin{thm}
For interval uncertainty, the solution to the nominal problem with midpoint costs is a $2$-approximation algorithm for the deterministic minmax regret problem.
\end{thm}
\begin{proof}
The linear program listed above is the dual of the problem solved by the optimizing player in the randomized framework. By weak duality, any feasible solution to the above program gives a lower bound on the value of the game in the randomized framework, $Z_{\mathrm{R}}$. The construction described above using costs $c^A$ and $c^{\overline{A}}$ gives a feasible solution to the program, and Lemma \ref{lmaam} allows us to express the resulting objective value in terms of the maximum deterministic regret for a solution set $M$ using $A=M$. Specifically,
\bea
\frac{Z_{\mathrm{D}}}{2} &\le& \frac{R_{\max}(M)}{2} \le Z_{\mathrm{R}} \le Z_{\mathrm{D}},
\eea
where the first inequality follows from the definition of the deterministic minmax regret problem, the second inequality follows using Lemma \ref{lmaam} with the feasible linear programming solution, and the third inequality follows from Theorem \ref{zrdinterval}.
\end{proof}
The potential gain from using randomization under interval uncertainty is not as significant as with discrete scenario uncertainty.
\begin{coro}
For interval uncertainty,
\bea
Z_{\mathrm{R}} \ge \frac{Z_{\mathrm{D}}}{2}.
\eea \label{ig2}
\end{coro}
\nin  Independently of the nominal problem, the integrality gap for the minmax regret problem is equal to $2$. A tight example for the corollary is easily constructed. Consider a problem with two items $E = \{e_1, e_2\}$ where the optimizing player must choose one item. Let $(c_e^+, c_e^-) = (0,1)$ for both items $e = e_1, e_2$. It can be verified that $Z_\mathrm{D} = 1$ and $Z_\mathrm{R} = 1/2$.

\section{General Uncertainty Sets}
\label{nphardsec}
In this section, we show that if the uncertainty set $\mc{C}$ is allowed to be a general nonnegative convex set and the nominal problem is polynomial solvable, the maximum expected regret problem becomes \NP-hard. Note that the deterministic maximum regret problem is a special case of the maximum expected regret problem. The result of this section thus implies that both randomized and deterministic minmax regret problems are \NP-hard under general convex uncertainty, even if the nominal problem is polynomial solvable.

We restate the maximum expected regret problem for general uncertainty sets. For a given marginal probability vector $p$, the maximum expected regret problem is, starting with the first line of \eqref{pause},
\begin{flalign}
\overline{R}_{\max}(p)&=  \max_{c \in \mc{C}} \left (\sum_{e \in E} c_e p_e - F^*(c) \right) \n \\
&=\max_{c \in \mc{C}} \left (\sum_{e \in E} c_e p_e - \min_{x \in \mc{X}} \sum_{e \in E} c_e x_e \right)\n \\
&=\max_{c \in \mc{C}} \max_{x \in \mc{X}} \left (\sum_{e \in E} c_e ( p_e - x_e) \right).
\end{flalign}
Negating the objective function, the maximum expected regret problem is equivalent to
\bea
-\overline{R}_{\max}(p) = \min_{c \in \mc{C}} \min_{x \in \mc{X}} \sum_{e \in E} c_e ( x_e - p_e), \label{negrmax}
\eea
for a given $p \in \ch(\mc{X})$.

Before we examine the complexity of \eqref{negrmax}, we consider the following problem, which we refer to as the \textit{bilinear combinatorial problem}:
\bea
\min_{c \in \mc{C}} \min_{x \in \mc{X}} \sum_{e \in E} c_e x_e.
\eea
We demonstrate the hardness of this problem via a reduction from the Hamiltonian path problem; the proof is similar to the standard proof for showing that the intersection of three matroids is \NP-hard. 
\begin{lma}
\label{lma:bilin}
For polynomial solvable nominal problems $F^*(c) = \min_{x \in \mc{X}} \sum_{e \in E} c_e x_e$ and nonnegative convex uncertainty sets $\mc{C}$, the bilinear combinatorial problem $\min_{c \in \mc{C}} \min_{x \in \mc{X}} \sum_{e \in E} c_e x_e$ is \NP-hard.
\end{lma}
\begin{proof}
Recall that the directed Hamiltonian path problem asks the following: given a directed graph $G=(V,E)$ with a designated source node $s$ and terminal node $t$, does there exist a path starting at $s$ and ending at $t$ that visits each node exactly once? For a given instance of the directed Hamiltonian path problem, we construct an instance of the bilinear combinatorial problem such that it has an optimal objective value of zero if an only if the graph contains a valid Hamiltonian path. 

We construct the set $\mc{C}$ to indicate the selection of edges such that each vertex has exactly one incoming edge (except for vertex $s$) and one outgoing edge (except for vertex $t$). Specifically, we say that an edge $e$ is selected if its cost $c_e$ is equal to zero, otherwise we refer to it as blocked. The notation $\delta^+(v)$ (respectively $\delta^-(v)$) indicates the set of outgoing (incoming) edges for vertex $v$. The constraints for the set $\mc{C}$ are
\bea
\label{Cconstraints}
\sum_{e \in \delta^-(v) } c_e &=& | \delta^-(v)| - 1, \quad v \in V \setminus \{s \}, \n \\
\sum_{e \in \delta^-(s) } c_e &=& | \delta^-(s)|, \n \\
\sum_{ e \in \delta^+(v) } c_e &=& | \delta^+(v)| - 1,\quad v \in V \setminus \{t \}, \n \\
\sum_{e \in \delta^+(t) } c_e &=& | \delta^+(t)|, \n \\
0 \le c_e \le 1, && \quad e \in E.
\eea
Note that for a given vertex, if one if its incoming edges is selected ($c_e = 0$), then all the remaining incoming edges must be blocked ($c_e = 1$); the same holds for outgoing edges. 

Define $\mc{X}$ to indicate the set of all feasible spanning trees for $G$, so that $\min_{x \in \mc{X}} \sum_{e \in E} c_e x_e$ is the minimum spanning tree problem. For an optimal solution $x^* = (x_1^*, \ldots, x_n^*)$ to the bilinear combinatorial problem giving zero objective value, the set $\{e | x_e^* = 0, e \in E\}$ indicates a valid Hamiltonian path, as the construction of $\mc{C}$ indicates that all vertices have one selected incoming and outgoing edge (except for $s$ and $t$), and the spanning tree ensures that no cycles are present. Finally, we have that the construction of the set $\mc{C}$ can be done in polynomial time.
\end{proof}
\begin{thm}
For polynomial solvable nominal problems $F^*(c) = \min_{x \in \mc{X}} \sum_{e \in E} c_e x_e$ and nonnegative convex uncertainty sets $\mc{C}$, the maximum expected regret problem $\max_{c \in \mc{C}} \left (\sum_{e \in E} c_e p_e - F^*(c) \right)$ where $p \in \ch(\mc{X})$ is \NP-hard.
\end{thm}
\begin{proof}
We modify the reduction used for Lemma \ref{lma:bilin} to account for the presence of some $p \in \ch(\mc{X})$ in the objective function of the maximum expected regret problem, which is now
\bea
\min_{c \in \mc{C}} \min_{x \in \mc{X}} \sum_{e \in E} c_e ( x_e - p_e).
\eea 
Again let $\mc{X}$ indicate the set of all feasible spanning trees for the directed graph $G = (V,E)$ and let $p \in \mc{X}$ be a valid spanning tree. We construct a new graph over the same set of vertices by taking $G$ and duplicating $|V|-1$ edges. For each edge given by the spanning tree $p$, we choose a corresponding edge in $G$ (note that there may be more than one option if both edges $(v_i, v_j)$ and $(v_j, v_i)$ are present, for example) and duplicate it. Let this new graph be denoted by $G' = (V, E')$, and let the set of all spanning trees over the new graph be indicated by $\mc{X}'$. Let $\wt{p} \in \mc{X}'$ indicate the set of edges that were constructed via duplication (i.e. the edges $E' \setminus E$), which is a valid spanning tree for $G'$. We finally construct the set $\mc{C}'$ using the inequalities in \eqref{Cconstraints} but over $E'$ instead of $E$.

Now consider the modified maximum expected regret problem
\bea
\min_{c \in \mc{C}'} \min_{x \in \mc{X}'} \sum_{e \in E'} c_e ( x_e - \wt{p}_e) = \min_{c \in \mc{C}'} \min_{x \in \mc{X}'}  \left ( \sum_{e \in E} c_e x_e + \sum_{e \in E' \setminus E} c_e (x_e-1) \right).
\eea
It can be seen that the modified problem has an objective value equal to $-(|V|-1)$ if and only if $G$ has a Hamiltonian path. This corresponds to the first sum in objective function being equal to zero, and the second sum being equal to $-(|V|-1)$. As before, for an optimal solution $x^* = (x_1^*, \ldots, x_{n+|V|-1}^*)$ to the modified problem, the set $\{e | x_e^* = 0, e \in E'\}$ gives a Hamiltonian path that is valid for both $G'$ and $G$. Notice that all of the duplicated edges $e \in E' \setminus E$ must be blocked ($c_e = 1$) and not selected by the minimum spanning tree ($x_e = 0$) for the objective value to be equal to $-(|V|-1)$. To finish the proof, we observe that the construction of $\mc{C'}$ and $G'$ can be accomplished in polynomial time.
\end{proof}

\section{Conclusion}
\label{conclusion}
We have shown that for both the interval and discrete scenario representations of uncertainty, the randomized minmax regret version of any polynomial solvable combinatorial problem is polynomial solvable. Furthermore, the maximum expected regret in the randomized model is upper bounded by the maximum regret of the deterministic model. These results, including the fact that there always exists a polynomial-sized optimal solution for randomized minmax regret, are at first glance somewhat surprising. Intuitively, the polynomial solvability of the randomized model results from the fact that a linear program must be solved instead of the integer program or mixed integer program (which is required for the deterministic model). The improvement in performance holds because the adversary has less power in the randomized model than the deterministic model.

For many applications that are not adversarial in nature, the randomized minmax regret criteria is likely a more appropriate model than the deterministic version. In particular, the deterministic solution may be overly conservative since costs are not truly chosen in an adversarial fashion in response to the selected solution. On the other hand, one must be willing to tolerate higher variance if randomization is used.

Our results on lower bounds for randomized minmax regret in relation to deterministic minmax regret, specifically Corollary \ref{igk} and \ref{ig2}, have important implications for approximating deterministic minmax regret problems. In Kasperski \cite{kasperski08}, it is posed as an open problem whether or not there exist approximation algorithms for interval uncertainty that, for some specific nominal problems, achieve an approximation ratio better than $2$. In some sense, we have answered this question in the negative. Corollary \ref{ig2} indicates that the integrality gap for the minmax regret problem is equal to $2$, and it is easy to create instances of nearly all nominal problems that achieve this gap. The same can be argued for the integrality gap of $k$ under discrete scenario uncertainty. Integrality gaps bound the best possible performance that can be obtained from approximation algorithms based on linear programming relaxations, which includes many approximation techniques \cite{williamson11}. Nonetheless, an important shortcoming of our primal-dual algorithms is that they do not use optimal dual solutions. It may be possible to design algorithms that perform better on specific problem instances by using optimal dual solutions. Under discrete scenario uncertainty, for instance, solving the nominal problem with costs averaged by dual variable weights, rather than averaged uniformly, may give improved performance.

An important future step with randomized minmax regret research is to develop approximation algorithms (now in the randomized model) for dealing with nominal problems that are already \NP-complete. This problem is non-trivial: an algorithm with an approximation factor $\alpha$ for a nominal problem does not immediately yield an algorithm to approximate the randomized minmax regret problem with a factor $\alpha$. Another interesting topic to study from an experimental perspective is a hybrid approach that employs both deterministic and randomized minmax regret. For example, one could find a solution that minimizes the maximum expected regret, subject to the maximum regret being no greater than some constant. The algorithm for this problem can be easily constructed by combining our results with existing work, but may no longer be polynomial solvable.

\bibliographystyle{plainnat}
\bibliography{rand_mmr_v2}

\end{document}